\documentclass[10pt]{iopart}

\usepackage{mathrsfs}
\usepackage{amssymb}
\usepackage{amsthm}

\usepackage{iopams}  
\usepackage{citesort}

\usepackage{color}

 \usepackage{graphics}
 \usepackage{graphicx}
 \usepackage{epsfig}

 \newtheorem{lemma}{Lemma}[section]
 \newtheorem{proposition}{Proposition}[section]
 \newtheorem{theorem}{Theorem}[section]
 \newtheorem{remark}{Remark}[section]
 \newtheorem{corollary}{Corollary}[section]

\newcommand{\ffd}{\overline{\overline{f}}}
\newcommand{\ff}{\overline{f}}

\begin{document}

\title[Freezing phase transition in a fractal potential]{Freezing phase transition in a fractal potential}

\author{Cesar Maldonado}

\address{IPICYT/Divisi\'{o}n de Matem\'{a}ticas Aplicadas. Camino a la Presa San Jos\'{e} 2055, Lomas 4a. secci\'{o}n, C.P. 78216, San Luis Potos\'{i}, S.L.P. Mexico. }

\ead{cesar.maldonado@ipicyt.edu.mx}

\author{R. Salgado-Garc\'{\i}a}

\address{Centro de Investigaci\'on en Ciencias - IICBA, Universidad Aut\'onoma del Estado de Morelos. Avenida Universidad 1001, colonia Chamilpa, C.P. 62209, Cuernavaca Morelos, Mexico.}

\ead{raulsg@uaem.mx}

%
%

\begin{abstract}

In this work we propose a simple example of a one-dimensional thermodynamic system where non-interacting particles are allowed to move over the $[0,1]$ interval, which are influenced by a potential with a fractal structure. We prove that the system exhibits a phase transition at a finte temperature, which is characterized by the fact that the Gibbs-Boltzmann probability measure passes from being absolutely continuous with respect to Lebesgue (at high temperature) to being singular continuous (at low temperatures). We prove that below the critical temperature (when the Gibbs-Boltzmann probability measure is singular continuous) the probability measure is supported on the middle-third Cantor set and that further lowering the temperature, the probability measure does not change anymore. This means that, in some sense, the system reaches the ground-state before the zero temperature, indicating  that the system ``freezes'' at a positive temperature.


\end{abstract}

%
%
%
\maketitle
%
%

\section{Introduction}
\label{sec:intro}

Thermodynamical systems having some underlying characteristics of fractality has been a subject of intense research~\cite{ben2000diffusion,bunde2013fractals,NovikovCh7Fractals2006,olemskoui1993application}. It has been noticed that the fractal nature of these class of systems (both, in equilibrium and out of equilibrium) affects, in several different ways the observed thermodynamic properties~\cite{arneodo1995thermodynamics,yang1987approximate,gefen1984phase,wang2003fractal}. Spin systems on fractal lattices~\cite{gefen1984phase,yang1987approximate}, the energy landscape of protein folding~\cite{lidar1999fractal,neusius2008subdiffusion,granek2011proteins}, the  effective thermal conductivity of liquids with nanoparticle inclusions~\cite{wang2003fractal} and the diffusion of particles on fractal geometries~\cite{sandev2017anomalous,barlow1998diffusions} are a few examples of thermodynamical systems having a underlying fractal structure exhibiting properties influenced by the fractality. 

In this paper we provide a simple example of a one-dimensional thermodynamic system with a fractal structure, which can be solved in an explicit way. This system consists of an ensemble of non-interacting particles moving on a one-dimensional fractal potential. It exhibits some features that, up to our knowledge, has not been  previously reported. For example, most of the models mentioned above assume that the particles move on a geometrically confined region with a fractal structure. In contrast, in our model, the particles can move on a non-fractal space, the interval $[0,1]$, but they are influenced by a potential with a fractal structure. We show that this system undergoes a phase transition, in which, below certain critical temperature; all the particles collapse into a fractal set. More precisely, the probability measure characterizing the equilibrium state (the Gibbs-Boltzmann probability measure) is absolutely continuous with respect to Lebesgue at high temperatures but, it is singular-continuous at low temperatures.

In this context, our work is related to the freezing phase transition phenomenon, which was studied in one-dimensional spin systems with long-range interactions~\cite{bruin2013renormalization} and is defined as the transition in which the support of the equilibrium state collapses into the ground-state~\footnote{The term ground-state in physics is the set of configurations minimizing the energy.} at a finite temperature. This means that, below the critical temperature the system does not change anymore as the temperature is lowered further (this is why this kind of phase transition is called ``freezing''~\cite{bruin2013renormalization}).  In contrast, in our model, the fractal potential can be defined in such a way that the set of configurations minimizing the energy is empty, implying that it has an empty ground-state. Despite this property, the system still exhibits the phase transition and the support of the Gibbs-Boltzmann probability measure still collapses into a fractal set (the middle-third Cantor set by construction). 

In this work we first show that our model for a fractal potential is well defined. We show that a slight modification of the potential model can be made it such that the minima of the fractal potential is just the middle-third Cantor set. We show that  the Gibbs-Boltzmann distribution associated to the fractal potential undergoes through a phase transition at a finite temperature. Subsequently we show that above the critical temperature the Gibbs-Boltzmann distribution is absolutely continuous with respect to Lebesgue and below, the Gibbs-Boltzmann distribution is singular-continuous, in fact it corresponds to the Cantor distribution. More over, the critical temperature is actually related to the fractal dimension of the Cantor set.

The paper is organized as follows: in Section \ref{sec:model} we give the concepts and definitions used in the paper. Section~\ref{sec:main} is devoted to the statements of our results. We give comments and some concluding remarks in section~\ref{sec:concluding}. And finally in sections~\ref{sec:proofs} and~\ref{sec:proof-lemmas}, we give the proofs for our results and the technical lemmas, respectively.

\section{The Model and Generalities}
\label{sec:model}

\subsection{Cantor structure of the model}
\label{subsec:Cantor}

First of all, let us give the standard definition of the middle-third  Cantor set as the limit of a monotonically decreasing sequence of finite union of intervals.  Let $C_{0}:=[0,1]$ be the first element of the sequence, and define $C_{1}:=[0,1/3]\cup[2/3,1]$, obtained by removing the middle third interval (see Figure~\ref{fig:CnBn}), and so on. Then, the middle-third Cantor set can be defined as the limit of the sequence of $C_{m}$, as follows, 
\begin{equation*}
C := \bigcap_{m=1}^\infty C_m.
\end{equation*}
By convenience, for every $m\in\mathbb{N}$, let us write the removed middle-third intervals as follows,
\begin{equation*}
B_m := C_{m-1}\setminus C_{m}. 
\end{equation*}
In Figure~\ref{fig:CnBn} we can see a representation of $C_n$ and $B_n$ for the first values of $n$. 
\begin{figure}[ht]
\begin{center}
\scalebox{0.4}{\includegraphics{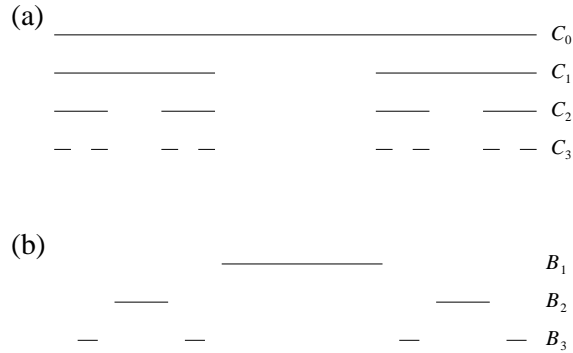}}
\end{center}
     \caption{
Schematic representation of $C_n$ and $B_n$ for the first values of $n$.
              }
\label{fig:CnBn}
\end{figure}
%

\subsection{The potential}
\label{subsec:potential}
We propose a model for a fractal potential. In order to define it, let us consider a function $\phi_k : [0,1] \to \mathbb{R}$ defined as follows:
\begin{equation*}
\phi_{k}(x) := \left\{
 \begin{array}{cc}
   1 & \mbox{if } x \in  B_{k},\\
   0 & \mbox{if } x \not\in B_k.
 \end{array} \right.
\end{equation*}
Notice that the function $\phi_k(x)$ is nothing but the characteristic function of the set $B_k$.
Next, let the function $\psi_n : [0,1] \to \mathbb{R}$ be given by
\begin{equation}
\label{eq:psi_n}
\psi_n (x) := - \sum_{k=1}^{n} A_k \phi_{k}(x), 
\end{equation}
where $\{A_k : k\in \mathbb{N}\}$ is a monotonically increasing sequence such that $A_k >0$ for every $k$ and $\lim_{k\to \infty} A_k = +\infty$. 
The fractal potential $\psi : [0,1] \to \mathbb{R}$ is then defined as the limit of the sequence of functions $\psi_n$ as $n\to \infty$,
\[
\psi (x) := \lim_{n\to \infty}\psi_{n}(x) = - \sum_{k=1}^{\infty} A_k \phi_{k}(x). 
\]
Notice that the sequence of finite-step potentials depends on an increasing sequence of real numbers which is not bounded, by definition. In figure~\ref{fig:potential} we can appreciate a schematic representation of the potential function $\psi_n$ for a finite value of $n$. Then, the first step is to determine whether the limit function $\psi (x)$ exists in some sense. This is actually the matter of our first result, the existence of the limit with respect to the $L^{p}$-norm.

\begin{proposition}\label{prop:conver-potential}

Consider the finite-step potential $\psi_n$ defined in equation~(\ref{eq:psi_n}).
If the sequence of constants $\{ A_k : k\in \mathbb{N} \}$ is such that $A_k = o\big(\left(\frac{3}{2}\right)^{k/p}\big)$, for every $p\geq1$, then the sequence $\{\psi_n : n\in\mathbb{N}\}$ converges in the $L^{p}$-norm.
\end{proposition}  
The proof of this proposition can be found in section~\ref{proof-prop:potential}. 
\begin{figure}[ht]
\begin{center}
\scalebox{0.35}{\includegraphics{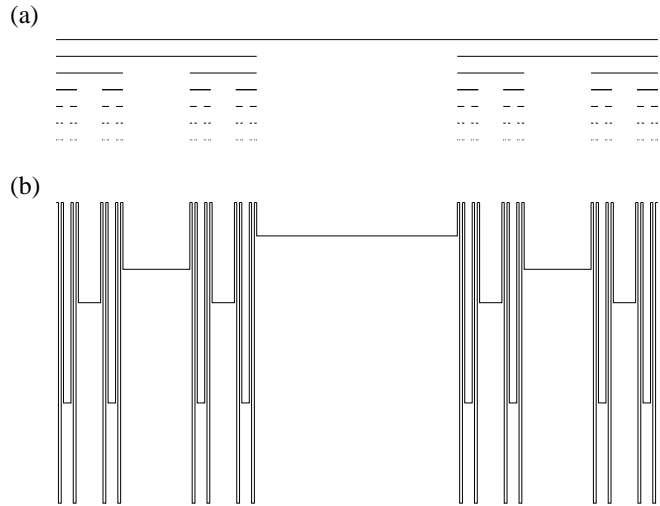}}
\end{center}
     \caption{
Schematic representation $\psi_n$. (a) The Cantor set construction. (b) The potential for a finite $n$ (actually for $n=5$). Notice that the value of the potencial on $B_k$ ($1\leq k \leq n$) is constant and that the value of the potential on $C_n$ is zero.
              }
\label{fig:potential}
\end{figure}
%

Notice that the above defined potential is such that the set of points minimizing $\psi$ is empty. This is because the points minimizing the $n$th potential, $\psi_n$, is just $B_n$ by construction (this is because $\psi_n(x) = -A_n$ if $x\in B_n$, also note that $\psi_n (x) = 0 $ for all $x\in C_n$). This means that the $n$th groundstate is the set $B_n$, and we can see that $\lim_{n\to\infty }B_n = \emptyset $. However, we can redefine the potential in order to have a non empty limiting groundstate. This is done by requiring that the value of the potential in $ C_n$ be the lowest. We denote by $\tilde \psi$ the potential
\begin{equation*}
\tilde \psi_n (x) := - \sum_{k=1}^{n} A_k \phi_{k}(x) - A_{n+1} \chi_{C_{n}}.
\end{equation*}
This slight modification makes the groundstate for $\tilde \psi_n$ to be $C_n$. Then, it is clear that the limiting groundstate becomes the middle-third Cantor set by construction. It is easy to see that $\tilde \psi$ is well defined and that all the theorems we state below also hold for $\tilde \psi$.

\subsection{Gibbs-Boltzmann distribution}
\label{subsec:Gibbs}
Here we introduce the main quantity of this paper, which is the distribution associated to the potentials defined in the previous section. We consider our space $[0,1]$ and with the collection of intervals as the measurable sets. Given a real positive number $\beta$ as the inverse temperature of the system, and a natural $n$, let us define $\mu_{\beta,n}$, the Gibbs-Boltzmann distribution associated to the finite-step potential $\psi_{n}$ of a measurable set $A\subset[0,1]$, as follows:
\begin{equation}
\label{eq:Gibbs}
\mu_{\beta,n }(A) = \int_{A} \rho_{\beta,n}(x) dx,
\end{equation}
where $\rho_{\beta, n}$ is the density of $\mu_{\beta,n}$ defined as the Boltzmann factor,
\[
\rho_{\beta,n}(x) := \frac{e^{-\beta \psi_{n}(x)}}{Z_{\beta,n}}.
\]
And $Z_{\beta,n}$ is the partition function given by
\[
Z_{\beta,n}:= \int_{[0,1]} e^{-\beta\psi_{n}(x)}dx.
\]
Our main result establishes that this distribution has a support that ``collapses'' into the middle-third Cantor set at a finite temperature. This means that at a finite temperature the (non-interacting) particles reach the groundstate before zero temperature.

\subsection{Weak distance and the Kullback-Leibler divergence}
\label{subsec:distance}

We are interested in the convergence of a sequence of probability distributions  $\mu_{\beta, n}$  to certain limiting measure. For this purpose we introduce the notions of two different forms in which the measures may converge. First, on the one hand, we consider the weak distance between two measures $\mu$ and $\nu$ defined on the same probability space,
\[
d(\mu,\nu) := \sup  \left\{\bigg|  \int f d\mu - \int f d\nu \bigg|  : f\in C^0  \right\}. 
\]
On the other hand, one has a stronger sense of convergence, called total variation. Given two probability distributions $\mu$ and $\nu$ on a sigma-algebra $\mathcal{F}$, the total variation distance is equal to
\[
||\mu -\nu||_{TV} : = \sup_{A\in\mathcal{F}}|\mu(A)-\nu(A)|.
\]
And we also have the Kulback-Leibler divergence for two probability distributions $\mu$ and $\nu$, which is given by
\[
D_{KL}(\mu||\nu) = \int\log\frac{d\mu}{d\nu}, 
\]
whenever $\mu<<\nu$ and infinite otherwise. It is important to mention that $D_{KL}(\mu||\nu)$  is not a distance on the space of probability distributions, yet it is related to the total variation distance by the Pinsker's inequality given by,
\[
\frac{1}{2}||\mu-\nu||_{TV}^{2}\leq D_{KL}(\mu||\nu).
\]
The above inequality states that convergence in Kullback-Leibler sense implies convergence in total variation (see for instance Ref.~\cite{Harremoes2009}).

\section{Results}
\label{sec:main}

Our results state that the Gibbs-Boltzmann measure for the fractal potential $\psi$ is well defined at all temperatures, in general, but there exists a critical inverse temperature $\beta_{c}$ for which the limiting measure and the convergence properties of the sequence $\mu_{\beta,n}$ change substantially. That is what we refer as freezing phase transition. The following proposition claims the existence of the limiting measure.

\begin{proposition}
\label{prop:existence}
Let $\mu_{\beta, n}$ be the Gibbs-Boltzmann distribution associated to the finite-step potential $\psi_{n}$ at the inverse temperature $\beta$ given by equation (\ref{eq:Gibbs}). Consider an arbitrary but fixed real number $\alpha>0$, $k\in\mathbb{N}$ and let $A_{k} :=\alpha k$. Then the sequence $\{ \mu_{\beta,n} : n\in\mathbb{N}\}$ converges to a unique limit measure $\mu_\beta$ for all $\beta > 0$. 
\end{proposition}

\begin{proof}
The proof of this proposition is actually contained in theorems \ref{theorem:beta<beta_c} and \ref{theorem:Cantor-dist} below.
\end{proof}

Before we continue with  our results, let us make the following remark on the partition function and the existence of the critical inverse temperature.

\begin{remark}
\label{prop:Z}
The existence of the $\beta_{c}$ in our system is a consequence of the role that fractality plays on the different sets of the space. This is why it is related with the fractal dimension of the Cantor set. In order to see it, let us compute the partition function, take any $\beta$ and $n\in\mathbb{N}$,
\begin{eqnarray*}
Z_{\beta,n} &=& \int_{[0,1]} e^{-\beta \psi_n(x)} dx = \sum_{k=1}^n \int_{B_k} e^{\beta \alpha k} dx + \int_{C_n} dx
\\
&=&   \sum_{k=1}^n \ell ( B_k ) e^{\beta \alpha k} + \ell(C_n).
\end{eqnarray*}
These equalities hold because of the definition of $\psi_{n}$, and where, $\ell(A)$ stands for the length of the set $A\subset[0,1]$. Also, observe that we can always take that partition of the unit interval, say, $(\bigcup_{k}B_{k})\cup C_{n}$, which is induced by the chosen $n$. Now observe that $\ell(C_{n}) = \big(\frac{2}{3}\big)^{n}$ and $\ell(B_{k}) = \frac{2^{k-1}}{3^{k}}$. So let us introduce the quantity, 
\[
h = -\log\Big(\frac{2}{3}\Big),
\]
then, one has that:
\begin{equation}
\label{eq:Partition}
Z_{\beta,n} = \frac{1}{2} \sum_{k=1}^n e^{(\beta \alpha -h) k} + e^{-hn}.
\end{equation}
The last equation  makes clear that for $\beta\alpha - h\leq 0$ the sequence $Z_{\beta,n}$ converges and on the contrary, for $\beta\alpha - h> 0$ it diverges. So we define our critical inverse temperature $\beta_{c} = \frac{h}{\alpha} =\frac{\log(2/3)}{\alpha}$.
\end{remark}

\subsection{Case $\beta\leq\beta_{c}$.}

Here we give our result for the case when the inverse temperature is lower that the critical one. 

\begin{theorem}
\label{theorem:beta<beta_c}
For $\beta < \beta_{c}$ the sequence of measures $\{ \mu_{\beta,n} : n\in\mathbb{N}\}$ converges in the total variation distance to a limit measure which is absolutely continuous with respect to Lebesgue. Moreover, its density $\rho_{\beta} := \lim_{n\to \infty} \rho_{\beta,n}$, is given as follows:
\[
\rho_{\beta}(x) = \frac{e^{\psi(x)}}{Z_{\beta}},
\]
where $\psi(x)$ is the fractal potential assured by Proposition~\ref{prop:conver-potential} and $Z_{\beta}$ is the partition function given by,
\[
Z_{\beta} = \frac{1}{2}\frac{1}{1-e^{\beta \alpha-h}}.
\]
\end{theorem}

The proof of this theorem is given in section~\ref{proof:beta<beta_c}.

\subsection{Case $\beta>\beta_{c}$.}

Next, for the case where the inverse temperature is larger that the critical one, stated in the proposition~\ref{prop:Z}, one has the following results.

\begin{theorem}
\label{theorem:Cantor-dist}
For $\beta>\beta_{c}$ one has that the sequence $\{\mu_{\beta, n}\}$ converges to a limit measure $\mu_{*}$ in the weak distance. Moreover the limit measure corresponds to the Cantor distribution, which is a singular continuous distribution, such that
$
\mu_{*} (C) = 1.
$
\end{theorem}

The proof of this theorem is given in section~\ref{proofTheoCantorDist}. Next, we give other results related to the cumulative distribution function of the limit measure $\mu_{*}$.

\begin{proposition}
\label{prop:continuity}
For $\beta > \beta_c$ the cumulative distribution function $F(x)$ associated to the measure 
$\mu_*$ is continuous. 
\end{proposition}

\begin{proof} Take $x,y \in [0,1]$ such that $|x-y| \leq 3^{-k}$ for some fixed $k \in \mathbb{N}$. 
Then we have that i) both $x$ and $y$ are contained in some interval of $C_k$, ii)  both $x$ and $y$ 
are contained in some interval of $B_l$ for some $ 1 \leq l \leq k$, iii) one of the points $x,y$ belongs to
some interval of $C_k$ and the other belongs to some  interval of $B_l$ for some $ 1 \leq l \leq k$. 

For the case i) let us denote by $C_{k,j}$ the interval containing $x$ and $y$. The notice that  
\[
\vert F_{\beta,n}(x) -   F_{\beta,n}(y)  \vert =  \left\vert \int_{x}^y \rho_{\beta,n} (z) dz \right\vert \leq  \mu_{\beta,n} (C_{k,j}). 
\]
Since
\[
\lim_{n\to \infty} \mu_{\beta,n} (C_k) = 1,
\]
by Theorem~\ref{theorem:beta<beta_c}, we have that 
\[
\lim_{n\to \infty} \vert F_{\beta,n}(y) -   F_{\beta,n}(x)  \vert  \leq \lim_{n\to \infty} \mu_{\beta,n} (C_{k,j}) = 2^{-k}. 
\]

For the case ii) let us denote by $B_{l,i}$ the interval containing $x$ and $y$. Then we have that
\[
\vert F_{\beta,n}(x) -   F_{\beta,n}(y)  \vert =  \left\vert \int_{x}^y \rho_{\beta,n} (z) dz \right\vert  \leq \mu_{\beta,n} (B_{l,i}). 
\]
By Theorem~\ref{theorem:beta<beta_c} it is clear that
\[
\lim_{n\to\infty} \vert F_{\beta,n}(x) -   F_{\beta,n}(y)  \vert  \leq  \lim_{n\to\infty}   \mu_{\beta,n} (B_{l,i}) =0. 
\]
Finally for the case iii) it is easy to see that
\[
\vert F_{\beta,n}(x) -   F_{\beta,n}(y)  \vert =  \left\vert \int_{x}^y \rho_{\beta,n} (z) dz \right\vert  \leq  \mu_{\beta,n} (C_{k,j}) + \mu_{\beta,n} (B_{l,i}). 
\]
Again, by Theorem~\ref{theorem:beta<beta_c}, we have that 
\[
\lim_{n\to \infty} \vert F_{\beta,n}(y) -   F_{\beta,n}(x)  \vert  \leq \lim_{n\to \infty} \left[ \mu_{\beta,n} (C_{k,j}) + \mu_{\beta,n} (B_{l,i}) \right] =  2^{-k}. 
\]
This shows that the  cumulative distribution function $F(x)$ is continuous. 
\end{proof}

\begin{proposition}
\label{prop:F=0}
For $\beta > \beta_c$ the cumulative distribution function $F_{\beta,n}(x)$ associated to
the density $\rho_{\beta,n}$ is such that
\[
\lim_{n\to \infty} \frac{dF_{\beta,n}(x)}{dx} = 0,
\]
for almost every $x$ with respect to Lebesgue. 
\end{proposition}

\begin{proof} The derivative of the cumulative distribution function $F_{\beta,n}(x)$ is the density $\rho_{\beta,n}$. This means that we only need to prove that $\rho_{\beta,n}(x) = 0$ for almost every $x$ with respect to Lebesgue. Thus, take $x\in B_k $ for fixed $k$ and notice that
\[
\rho_{\beta,n}(x) = \frac{e^{-\beta \psi_n(x)} }{Z_{\beta,n}} = \frac{e^{\beta \alpha k}}{ Z_{\beta,n}}.
\]
Now, by Remark~\ref{prop:Z} we have that $Z_{\beta,n} \to \infty$ as $n\to \infty$. Since we assumed that $k$ is finite it is clear that
\[
\rho_{\beta,n}(x) = \frac{e^{\beta \alpha k}}{ Z_{\beta,n}} \to 0 \qquad  \mbox{ as } \quad n\to \infty.
\]
Since $\bigcup_{k=1}^\infty B_k = [0,1]\setminus C $ we have that $\rho_{\beta,n}(x) = 0$ for almost every $x$ with respect to Lebesgue. This proves our statement. 
\end{proof}

\section{Concluding remarks}
\label{sec:concluding}

We have introduced a simple model which can be dealt with in a rigorous way. This model is a system of non-interacting particles which are influenced by a fractal potential. The system exhibits a phase transition, which is mainly characterized by the fact that the support of the Gibbs-Boltzmann measure ``collapses'' into the middle-third Cantor set. The latter essentially means that above  certain critical temperature, the  Gibbs-Boltzmann  probability measure is absolutely continuous with respect to Lebesgue and below the critical temperature it becomes singular continuous. Moreover, once the system has reached the critical temperature, the  Gibbs-Boltzmann measure does not change anymore by further decreasing the temperature. This last property means that the system ``freezes'' at a finite temperature, a phenomenon that was already shown to occur in spin systems~\cite{bruin2013renormalization}.

From the physical point of view, the occurrence of phase transition in our model is related to the occurrence of anomalous diffusion. For example, we can think of our model as a system consisting of non-interacting overdamped particles moving on a fractal periodic potential by extending our fractal potential model periodically to the whole real line. In this case, it is known that the exact diffusion coefficient is proportional to the inverse of the partition function~\cite{festa1978diffusion},
\[
D_{\mathrm{eff}}  \propto \frac{1}{Z_\beta Z_{-\beta} }.
\]
Since $Z_\beta$ diverges for $\beta > \beta_c$ it is clear that $D_{\mathrm{eff}}$ goes to zero, having as a consequence the occurrence of a normal to anomalous diffusion transition.  This fact suggest that the freezing transition is also accompanied by transition in which the particles does not spread linearly in time, a phenomenon known as \emph{slowing down}. The importance of our study is that we have shown that within the anomalous diffusion phase, the system does not present a density of particles because of the singular continuous measure. It would be interesting to see if some other systems, for which the occurrence of anomalous diffusion is reported (see for instance Refs.~\cite{sanders2006occurrence,metzler2014anomalous,armstead2003anomalous,salgado2013normal,salgado2016normal,hidalgo2017scarce}), exhibit a singular continuous measure. This, of course, would requiere a much more detailed further study.   

From the mathematical our system presents interesting properties. For example, we have shown that the convergence properties of the successive finite-step Gibbs-Boltzmann measures to the limit, substantially changes from one case to the other. That is, in one case the convergence is in total variation (or information theoretic sense) and below the critical temperature, it is only in the weak sense. Moreover, we have also shown that, despite the system has not a ground-state (the set of point minimizing the potential) our system still exhibits a freezing transition. 

And finally, we would like to point out the reference~\cite{vanEnter2007}, which in some sense is related to our work. In that paper, the authors propose a spin system with nearest-neighbor interaction. The potential is defined in a very special way, in a form of wells-in-wells, alternating one ferromagnetic and the other antiferromagnetic. They showed that there is no low-temperature limit of any sequence of Gibbs measures. It would be interesting to know if the phenomenon we observe in our system can also occur in spin systems.

\section{Proofs}
\label{sec:proofs}

\subsection{Proof of proposition~\ref{prop:conver-potential}}\label{proof-prop:potential}

By the definition of $\psi_{n}$, the sequence $\{\psi_{n}\}$ converges (probably to infinity) pointwise to the function $\psi$.
Here we go further by proving the convergence in $L^{p}$-norm. For exposition purposes, first we will show that the $L^{1}$-distance between the finite step functions $\psi_{r}$ and $\psi_{s}$ goes to zero if one takes $r$ and $s$ large enough. 
Let $r<s$, for every $x\in[0,1]$ we have that,
\[
| \psi_{s}(x) - \psi_{r}(x)| = \sum_{k=r+1}^{s} A_{k}\phi_{k}(x).
\]
Observe that, for every $n\in\mathbb{N}$, one has the partition for the unit interval given as follows, 
\begin{equation}
\label{eq:partition}
[0,1] = \Big( \bigcup_{i=1}^{n} B_{i} \Big)\cup C_{n}.
\end{equation}
Consider the $L^{1}$-distance between $\psi_{r}$ and $\psi_{s}$, and using the partition (\ref{eq:partition}), one has
\begin{eqnarray}\label{eq:int-partition}
\hspace{-1cm}
\nonumber  || \psi_{s} - \psi_{r}||_{L^{1}} &= & \sum_{i=1}^{n}\int_{B_{i}}|\psi_{s}(x) - \psi_{r}(x)|dx + \int_{C_{n}}| \psi_{s}(x) - \psi_{r}(x)| dx\\
& = & \sum_{i=1}^{n}\int_{B_{i}}\sum_{k=r+1}^{s} A_{k}\phi_{k}(x) dx + \int_{C_{n}}\sum_{k=r+1}^{s} A_{k}\phi_{k}(x) dx.
\end{eqnarray}
\noindent
Recall that $\phi_{k}$ is the characteristic function of the set $B_{k}$, and observe that for every $n$, the first term in equation (\ref{eq:int-partition}), can be written as follows:
\begin{eqnarray*}
\sum_{i=1}^{n}\int_{B_{i}}\sum_{k=r+1}^{s} A_{k}\phi_{k}(x) dx & =& \sum_{i=1}^{n}\sum_{k=r+1}^{s}\int_{[0,1]} A_{k}\cdot\chi_{B_{i}\cap B_{k}}(x) dx.
\end{eqnarray*}
A direct but useful observation is that, whenever $i\neq k$,
\begin{equation}\label{observation1}
B_{i} \cap B_{k} = \emptyset.
\end{equation}
\noindent
Next, the second term in eq.~(\ref{eq:int-partition}) can be written as follows, 
\begin{eqnarray*}
\int_{C_{n}}\sum_{k=r+1}^{s} A_{k}\phi_{k}(x) dx 
&=& \sum_{k=r+1}^{s}\int_{[0,1]}A_{k}\cdot\chi_{C_{n}\cap B_{k}}(x)dx.
\end{eqnarray*}
\noindent
Now we proceed by proving that for every $n$, that is no matter the thickness of the partition one has the same estimate for the $L^{1}$-norm. For clarity, we do it separately for each of the two summands in eq. (\ref{eq:int-partition}) and we consider three cases, first, for $n\leq r$. Observe, that $C_{n}\cap B_{i} = \emptyset$ for all $i\leq n$, so one has that
\[
\sum_{i=1}^{n}\int_{[0,1]}\sum_{k=r+1}^{s}A_{k}\cdot \chi_{B_{i}\cap B_{k}}(x)dx = 0,
\]
using also observation (\ref{observation1}). For $r<n\leq s$ one has the following,
\[
\sum_{i=1}^{n}\int_{[0,1]}\sum_{k=r+1}^{s} A_{k}\cdot \chi_{B_{i}\cap B_{k}}(x)dx = \sum_{i=r+1}^{n}A_{i}\cdot\ell(B_{i}),
\]
where, $\ell$ stands for the usual length. Finally for $n>s$, one has that
\[
\sum_{i=1}^{n}\int_{[0,1]}\sum_{k=r+1}^{s} A_{k} \chi_{B_{i}\cap B_{k}}(x)dx = \sum_{i=r+1}^{s}A_{i}\ell( B_{i}),
\]
since for all $i>r$ the integral vanishes and using again the observation (\ref{observation1}). It remains to estimate the second term in (\ref{eq:int-partition}). 
We consider again the three cases for $n$, we remind the reader the observation $C_{n}\cap B_{k} = \emptyset$, for all $k\leq n$ and that $C_{n}\cap B_{k} = B_{k}$ whenever $k>n$. So, for $n\geq s$,
\[
\int_{C_{n}}\sum_{k=r+1}^{s} A_{k}\phi_{k}(x) dx = \int_{[0,1]}\sum_{k=r+1}^{n}A_{k} \chi_{C_{n}\cap B_{k}}(x)dx = 0.
\]
 For $r+1 \leq n < s$, one has
 \[
\int_{C_{n}}\sum_{k=r+1}^{s} A_{k}\phi_{k}(x) dx = \sum_{i=n+1}^{s}A_{i}\ell(B_{i}). 
 \]
And finally, for $n<r+1$,
\[
\int_{C_{n}}\sum_{k=r+1}^{s} A_{k}\phi_{k}(x) dx = \sum_{i=r+1}^{s}A_{i}\ell(B_{i}). 
\]
So putting all together, one has that for every $n$, 
\begin{eqnarray*}
||\psi_{s} - \psi_{r}||_{L^{1}} 
&=& \sum_{i=r+1}^{s}A_{i}\ell(B_{i}).
\end{eqnarray*}
Provided that, for every $i$, the length $\ell(B_{i}) = \frac{1}{2}\big(\frac{2}{3}\big)^{i}$, thus, making $A_{i} = o\big(\left(\frac{3}{2}\right)^{i}\big)$, one has that the summation $\sum_{i=r+1}^{s}A_{i}\ell(B_{i})$ goes to zero when $r,s$ are sufficiently large. That finishes the proof for the $L^{1}$-norm. 
Next, for $p>1$, we proceed analogously. By now, we will focus on the term,
\begin{eqnarray*}
| \psi_{s}(x)-\psi_{r}(x)|^{p} &=&\Big( \sum_{k=r+1}^{s} A_{k}\phi_{k}(x)\Big)^{p}\\
&=& \sum_{j_{1}+\cdots+j_{s-r}=p}{ {p}\choose {j_{1},\ldots,j_{s-r}} }\prod_{t=1}^{s-r}(A_{t}\phi_{t}(x))^{j_{t}}.
\end{eqnarray*}
Where the last part is the multinomial. Observe that independently of the combination of terms in the product, one has that $\prod_{t=1}^{s-r} \phi_{t}(x)^{j_{t}} = 1$ only when $j_{t} = p$ for every single $t$, and is 0 otherwise. This is true because of the the observation (\ref{observation1}). Therefore for every $x$ and every $p$,
\[
\Big(\sum_{k=r+1}^{s}A_{k}\phi_{k}(x)\Big)^{p} = \sum_{k=r+1}^{s}A_{k}^{p}\phi_{k}(x).
\]
Then, we proceed exactly as in the previous case, obtaining,
\[
||\psi_{s}-\psi_{r}||_{L^{p}} =\Big( \sum_{i=r+1}^{s}A_{i}^{p}\ell(B_{i})\Big)^{1/p}, 
\]
and so, choosing $A_{i}=o\big(\left(\frac{3}{2}\right)^{i/p}\big)$ allows us to obtain the desired result.
\qed


\subsection{Proof of Theorem~\ref{theorem:beta<beta_c}. }\label{proof:beta<beta_c} 

Now we proceed to prove that the sequence $\{ \mu_{\beta,n} : n\in\mathbb{N}\}$ 
converges to a limit measure in the information-theoretic sense by using the Kullback-Leibler divergence. 

First of all let us recall that, in the case of $\beta<\beta_c$, the partition function $Z_{\beta,n}$ given by (\ref{eq:Partition}) converges in the limit $n\to\infty$ to 
\[
\lim_{n\to\infty } Z_{\beta,n} = \frac{1}{2} \frac{1}{1-e^{\beta \alpha - h }}.
\]
And thus the density for the expected limit measure $\mu_{\beta}$ should be,
\[
\rho_{\beta} (x) = \frac{e^{-\beta \psi(x) }}{ Z_\beta},
\]
where $\psi$ is the fractal potential obtained as the limit of the sequence $\psi_{n}$, assured by Proposition~\ref{prop:conver-potential}.
Then our goal in the following, will be to prove that the sequence of  finite-step densities  $\{ \rho_{\beta,n} : n\in\mathbb{N}\}$ converges in the information-theoretic sense towards $\rho_{\beta}$, i.e., 
\[
D(\rho_{\beta,n} \Vert \rho_{\beta} ) \to 0\quad \mbox{as} \quad n\to\infty,
\]
where the Kullback-Leibler divergence $D (\cdot \Vert \cdot )$  between $\rho_{\beta,n}$ and  $\rho_{\beta}$ is given by
\[
D(\rho_{\beta,n} \Vert \rho_{\beta} )  = \int_{[0,1]} \rho_{\beta,n} (x) \log\left( \frac{\rho_{\beta,n}(x)}{\rho_{\beta} (x)}\right)dx.
\]
From the definition of the density of the finite-step measure $\mu_{\beta,n}$, we have that,
\begin{eqnarray}
D(\rho_{\beta,n} \Vert \rho_{\beta} )  &=& 
\int_{[0,1]} \frac{e^{-\beta \psi_n(x)}}{Z_{\beta,n}}  \log\left(\frac{e^{-\beta \psi_n(x)}}{Z_{\beta,n}} \frac{ Z_\beta}{e^{-\beta \psi(x) }}  \right)dx
\nonumber
\\
&=&
\log\left( \frac{ Z_\beta}{Z_{\beta,n}}  \right) 
+ \beta \int_{[0,1]} \left(\psi(x) -\psi_{n}(x) \right) \frac{e^{-\beta \psi_n(x)}}{Z_{\beta,n}}  dx.
\nonumber
\end{eqnarray}
And so, one obtains the following inequality,
\begin{eqnarray}
\hspace{-0.5cm}
D(\rho_{\beta,n} \Vert \rho_{\beta}) &\leq& \left\vert  \log\left( \frac{ Z_\beta}{Z_{\beta,n}}  \right) \right\vert  + \beta \int_{[0,1]} \left\vert \psi(x) -\psi_n(x) \right \vert \frac{e^{-\beta \psi_n(x)}}{Z_{\beta,n}}  dx.
\label{eq:inter-Dpsi-1}
\end{eqnarray}
Next, since $\phi_k$ is zero outside the set $B_k$, and $\psi(x)-\psi_{n}(x)$ is zero over $B_k $ for all $1\leq k \leq n$, we have that the second part of the right-hand side of inequality (\ref{eq:inter-Dpsi-1}) is equal to
\begin{eqnarray*}
\int_{[0,1]} \left\vert \psi(x) -\psi_n(x) \right \vert \frac{e^{-\beta \psi_n(x)}}{Z_{\beta,n}}  dx &=& 
\frac{1}{Z_{\beta,n}} \int_{C_n} \left\vert \psi(x) -\psi_n(x) \right \vert   dx 
\\
&\leq&
\frac{1}{Z_{\beta,n}} \Vert \psi(x) -\psi_n(x) \Vert_{L_1}.
\end{eqnarray*}
Then going back to~(\ref{eq:inter-Dpsi-1}) we can see that
\begin{eqnarray}
\fl
\qquad \qquad
 D(\rho_{\beta,n} \Vert \rho_{\beta} )  \leq \left\vert  \log\left( \frac{ Z_\beta}{Z_{\beta,n}}  \right) \right\vert  + \beta  \frac{1}{Z_{\beta,n}} \Vert \psi(x) -\psi_n(x) \Vert_{L_1}
\nonumber
\end{eqnarray}
\noindent
As we saw in Proposition~\ref{prop:conver-potential} we have that $\psi_n$ converges in $L_1$ to the limit potential $\psi$. Thus it is clear that $ \Vert \psi(x) -\psi_n(x) \Vert_{L_1} \to 0$ as $n\to\infty$. On the other hand we have also established that $Z_{\beta,n}$ goes to $Z_\beta$. 
Thus, we conclude that
\[
D(\rho_{\beta,n} \Vert \rho_{\beta} ) \to 0,
\]
as $n\to\infty$ which finishes the proof.
\qed


\subsection{Proof of Theorem~\ref{theorem:Cantor-dist}}\label{proofTheoCantorDist}
This section is devoted to the proof of Theorem~\ref{theorem:Cantor-dist}. We will prove that for every measurable Lipschitz continuous and bounded above function $f$, the sequence $\{\mu_{n}(f)\}$ is a Cauchy sequence with respect to the weak distance. Since the space of probability distributions with the weak distance is a complete metric space, the limit distribution $\mu_{*}$, actually exists and is a probability distribution. We will prove that given a measurable Lipschitz continuous and bounded above function $f:[0,1]\to\mathbb{R}$, for every $\varepsilon>0$ there exists a $N\in\mathbb{N}$ such that for $n, m>N$ one has that
\[
\bigg| \int_{[0,1]} fd\mu_{n} -\int_{[0,1]}fd\mu_{m}\bigg| < \varepsilon.
\]
\noindent
In order to do this, we give a list of useful lemmas that will be used in the proof of our theorem.

\begin{lemma}
\label{lemma:mass-concentration}
Let $\beta > \beta_c $, and $\mu_{\beta, n}$ be the Gibbs-Boltzmann measure associated to the $n$-step potential. Then for every $\epsilon >0$ and any $k\in \mathbb{N}$ there is a $N \in \mathbb{N}$ such that for all $n > N \geq k$, one has that,
\[
\mu_{\beta,n} (C_{k}) >1-\epsilon.
\] 
\end{lemma}

\begin{corollary} 
\label{corollary:mass-concentration}
Let $\beta > \beta_c $. For all $\epsilon >0$ there is a $N \in \mathbb{N}$ such that 
\begin{equation*}
\sum_{k=1}^L \mu_{\beta,m+L} (B_{k}) < \epsilon
\end{equation*}
for all $m > N$ and all $L\in \mathbb{N}$. 
\end{corollary}


\begin{lemma} 
\label{lemma:averages-f}
Let  $f : [0,1] \to \mathbb{R}$ be a bounded-above  Lipschitz continuous function. Then for every $\epsilon >0$ there exits a $N$ such that for all $n>m \geq N $ one has that,
\begin{equation*}
\bigg| \frac{1}{\ell (B_n)} \int_{B_n} f(x)dx - \frac{1}{\ell (B_m)} \int_{B_m} f(x)dx \bigg|  < \epsilon.
\end{equation*}
\end{lemma}

\begin{lemma}
\label{lemma:mass-concentration-bar}
For $\beta > \beta_c$ and for all $\epsilon > 0$ there is a $N \in \mathbb{N}$ such that 
\begin{equation*}
\sum_{l=0}^{L-1} \big| \mu_{\beta,n} \left( B_{n-l} \right)  -\mu_{\beta,m} \left( B_{m-l} \right)  \big| < \epsilon
\end{equation*}
for all $n,m >N$ and all $L $ such that $1\leq L \leq \min\{n,m\} $.

\end{lemma}

\begin{lemma}
\label{lemma:residual}

For all $\epsilon > 0$ there is a $N \in \mathbb{N}$ such that
\[
 \mu_{\beta,n}(C_n) < \epsilon 
\]
for all $n>N$.
\end{lemma}

\noindent 
Now, let us going on the proof of the theorem. First let us recall that we are in the case  $\beta>\beta_{c}$. According to the definition of $\mu_n$ we have that
\[
\int_{[0,1]} f d\mu_n = \int_{[0,1]} f(x) \frac{e^{-\beta \psi_n(x)}}{Z_n} dx. 
\]
Since $\psi$ is piecewise constant on the sets $B_k$ for $1\leq k \leq n$, it is natural to, and we can always, decompose the integral as a series by using the partition $\{B_k :  1\leq k \leq n \} \cup \{C_n\}$. This gives us
\begin{eqnarray}
\fl
\int_{[0,1]} f d\mu_n &=& \sum_{k=1}^n \int_{B_k} f(x) \frac{e^{-\beta \psi_n(x)}}{Z_n} dx + \int_{C_n} f(x) \frac{e^{-\beta \psi_n(x)}}{Z_n} dx
\nonumber
\\
\fl
&=& 
\sum_{k=1}^n \int_{B_k} f(x) \frac{e^{\beta A_k}}{Z_n} dx + \frac{1}{Z_n} \int_{C_n} f(x)  dx
\nonumber
\\
\fl
&=& 
\sum_{k=1}^n \left(  \frac{e^{\beta A_k}}{Z_n}\ell(B_k) \right) \frac{1}{\ell(B_k)}\int_{B_k} f(x)  dx + \left( \frac{\ell(C_n)}{Z_n} \right) \frac{1}{\ell(C_n)} \int_{C_n} f(x)  dx.
\nonumber
\end{eqnarray}
Where we used the fact that $\psi(x) = -A_k $ for $x\in B_k$, and $\psi(x) = 0$ elsewhere. Next, notice that the quantities inside the parenthesis  in the above equation can be written as follows:
\begin{eqnarray}
 \frac{e^{\beta A_k}}{Z_n}\ell(B_k) &=& \int_{B_k}  \frac{e^{\beta A_k}}{Z_n} dx =
 \int_{B_k}  \frac{e^{-\beta \psi_{n}(x)}}{Z_n} dx = \mu_{n}(B_k ),
\nonumber 
\\
 \frac{\ell(C_n)}{Z_n}  &=& \int_{C_n} \frac{1}{Z_n} dx  = \int_{C_n} \frac{e^{-\beta \psi_{n}(x)}}{Z_n} dx = \mu_n(C_n).
 \nonumber
 \end{eqnarray}
Then we have that
\begin{eqnarray}
\fl
\int_{[0,1]} f d\mu_n &=&
\sum_{k=1}^n \mu_{n}(B_k  ) \frac{1}{\ell(B_k)}\int_{B_k} f(x)  dx + 
 \mu_n(C_n) \frac{1}{\ell(C_n)} \int_{C_n} f(x)  dx. 
\label{eq:int-f-sum1}
\end{eqnarray}
In order to enlighten the calculations, let us introduce the notation $\ffd_k $ and $ \ff_n$ for the following averages:
\begin{equation*}
\ffd_k :=  \frac{1}{\ell(B_k)}\int_{B_k} f(x)  dx,
\ \hspace{1cm}
\ff_n  = \frac{1}{\ell(C_n)} \int_{C_n} f(x)  dx.
\end{equation*}
Using this short-hand notation we can rewrite expression~(\ref{eq:int-f-sum1}) as,
\begin{eqnarray}
\int_{[0,1]} f d\mu_n &=&
\sum_{k=1}^n \mu_{n}(B_k  ) \ffd_k + \mu_n(C_n) \ff_n. 
\label{eq:int-f-sum2}
\end{eqnarray}
If we assume that $n>m$ we can write $n = m + L$, with $L\geq 0$. Using this fact, we have that the first summand of the right hand side of the equation~(\ref{eq:int-f-sum2}) can be written as, 
\[
\hspace{-2.5cm}
\sum_{k=1}^n \mu_{n}(B_k  ) \ffd_k =  \sum_{k= 1}^{m+L} \mu_{m+L}(B_k  ) \ffd_k
=\sum_{k= 1}^{L} \mu_{m+L}(B_k  ) \ffd_k + \sum_{k= 1}^{m} \mu_{m+L}(B_{k+L}  ) \ffd_{k+L}.
\]
So, using an expression for $\int_{[0,1]} f d\mu_m$ analogous to (\ref{eq:int-f-sum2}), we write,
\begin{eqnarray*}
\fl
\int_{[0,1]} f d\mu_n - \int_{[0,1]} f d\mu_m =  
\sum_{k= 1}^{m} \mu_{m+L}(B_{k+L}  ) \ffd_{k+L} -\sum_{k=1}^m \mu_{m}(B_k  ) \ffd_k \\
\hspace{3cm}+ \sum_{k= 1}^{L} \mu_{m+L}(B_k  ) \ffd_k +  \mu_n(C_n) \ff_n - \mu_n(C_m) \ff_m.
\end{eqnarray*}
Hence, by adding and subtracting the quantity $\sum_{k=1}^{m}\mu_{m}(B_{k})\ffd_{k+L}$ and then using the triangle inequality, one obtains the following inequality,
\begin{eqnarray}
\fl
\bigg| \int_{[0,1]} f d\mu_n - \int_{[0,1]} f d\mu_m \bigg|  \leq  
\sum_{k= 1}^{m}  \big|  \mu_{m+L}(B_{k+L} ) - \mu_{m}(B_k  ) \big|
\big|   \ffd_{k+L}\big| +
\nonumber
\\ \hspace{-2cm}
\sum_{k= 1}^{m} \big| \ffd_{k+L}- \ffd_k \big| \mu_{m}(B_k  )
+ \sum_{k= 1}^{L}  \mu_{m+L}(B_k  ) \big|\ffd_k \big| +  \mu_n(C_n) \big| \ff_n \big| +  \mu_m(C_m) \big|\ff_m\big|.
\label{ineq:inter-diff-fs}
\end{eqnarray}
Now let us define the following quantities $\tilde{f}_1:= \max\{\big|   \ffd_{k+L}\big| : 1\leq k \leq m \}$, $\tilde{f}_2 := \max\{\big|   \ffd_{k}\big| : 1\leq k \leq m \}$ and  $M:= \mbox{max}\{ \mu_{m}(B_k  ) : 1\leq k \leq m\}$. So that, the inequality~(\ref{ineq:inter-diff-fs}) can be recast into,
\begin{eqnarray*}
\fl
\bigg| \int_{[0,1]} f d\mu_n - \int_{[0,1]} f d\mu_m \bigg|  &\leq&  
\tilde{f}_1 \sum_{k= 1}^{m}  \big|  \mu_{m+L}(B_{k+L} ) - \mu_{m}(B_k  ) \big|
+ M \sum_{k= 1}^{m} \big| \ffd_{k+L}- \ffd_k \big|
\\ 
&+& \tilde{f}_2 \sum_{k= 1}^{m}  \mu_{m+L}(B_k  )  +  \mu_n(C_n) \big| \ff_n \big| +  \mu_m(C_m) \big|\ff_m\big|.
\end{eqnarray*}

\noindent 
Then by Lemmas~\ref{lemma:mass-concentration}, \ref{lemma:averages-f}, \ref{lemma:mass-concentration-bar} and \ref{lemma:residual} we have that given any $\epsilon^\prime$  one can find a $N$ such that 
\begin{eqnarray*}
\hspace{1cm}
\fl
\bigg|  \int_{[0,1]} f d\mu_n - \int_{[0,1]} f d\mu_m \bigg|  &\leq&  
\left( \tilde{f}_1  + M +  \tilde{f}_2   +  \big| \ff_n \big| + \big|\ff_m\big|\right) \epsilon^\prime
\end{eqnarray*}
for all $n,m>N$. Where $N$ is the maximum out of the following set of integers,
\[
N := \max\{ N_1 (\epsilon^\prime), N_2 (\epsilon^\prime), N_3 (\epsilon^\prime),N_4 (\epsilon^\prime) \}, 
\]
the specific integers given in the lemmas ~\ref{lemma:mass-concentration}, \ref{lemma:averages-f}, \ref{lemma:mass-concentration-bar} and \ref{lemma:residual},
\begin{eqnarray}
N_1 (\epsilon^\prime) &=& \bigg| \frac{\log (\epsilon/J_\beta) }{\beta \alpha -h } \bigg|, \quad \hspace{0.5cm} \quad
N_2 (\epsilon^\prime) =\bigg\lfloor \bigg| \frac{  \log\left(\epsilon/6K_L\right)}{\log(3)} \bigg| \bigg\rfloor+1,
\nonumber
\\
N_3(\epsilon^\prime)  &=&  \bigg| \log \left( \frac{\epsilon {G_1}^2K_1 }{2 (K_1  + 2)}\right)\bigg|, \quad \quad
N_4 (\epsilon^\prime)= \bigg|\frac{\log \left( \epsilon K_1\right)}{\beta \alpha}\bigg|.
\nonumber
\end{eqnarray}
Finally by making $F := \max\{ \tilde{f}_1,  \tilde{f}_2,  \big| \ff_n \big|, \big|\ff_m\big| \} $, the maximum of the averages of $f$ and observing that $M \leq 1$, we have that
\begin{eqnarray*}
\bigg|  \int_{[0,1]} f d\mu_n - \int_{[0,1]} f d\mu_m \bigg|  &\leq&  \left( 1+ 4F \right) \epsilon^\prime.
\end{eqnarray*}
Thus, choosing  $\epsilon^\prime$ as $\epsilon/\left( 1+ 4F \right)$ we prove that the sequence $\{\mu_{n}(f)\}$ is a Cauchy sequence. It remains to prove that the measure is concentrated at the Cantor set, but this is a consequence of lemmas ~\ref{lemma:mass-concentration}, \ref{lemma:mass-concentration-bar} and \ref{lemma:residual}.

\qed

\section{Proofs of the technical lemmas}\label{sec:proof-lemmas}

\subsection{Proof of lemma~\ref{lemma:mass-concentration}.}

Recall that we are in the case where $\beta>\beta_{c}$. In order to estimate the lower bound for $\mu_{\beta,n}(C_k)$ for a given $C_{k}$ and $n$, observe that for $n>k$ one has that
\[
C_k = \left( \bigcup_{m=k+1}^n B_m  \right)\cup C_n,
\] 
which implies that $\{B_m \}_{m=k+1}^n \cup \{C_n\}$ is a partition of $C_k$. Then, $\mu_{\beta,n}(C_{k})$, can be written as follows,
\begin{equation*}
\mu_{\beta,n}(C_k) = \sum_{m=k+1}^n  \int_{B_m} \frac{e^{-\beta \psi_n(x)} }{Z_{\beta,n}} dx + \int_{C_n}\frac{e^{-\beta \psi_n(x)} }{Z_{\beta,n}} dx.
\end{equation*}
Since $\psi_n(x)$ is constant on each element of the partition, $\psi_{n}(x) = A_m$ for all $x\in B_m$ and that $\psi_{n}(x) = 0$ for all $x\in C_n$, then, it is clear that
\begin{equation*}
\mu_{\beta,n}(C_k) = \frac{1}{Z_{\beta,n}} \sum_{m=k+1}^n  e^{\beta A_m} \ell(B_m)   + \frac{ \ell(C_n)}{Z_{\beta,n}}.
\end{equation*}
Recall that $h$ was defined so that,
$\ell( B_m) = \frac{1}{2}e^{-hm}$, and $\ell (C_n) = e^{-hn}$. Also recall that $A_k =  \alpha k $, then we have 
\begin{equation*}
\mu_{\beta,n}(C_k) = \frac{1}{Z_{\beta,n}}\frac{1}{2} \sum_{m=k+1}^n  e^{(\beta \alpha - h)m}  + \frac{ e^{-hn}}{Z_{\beta,n}}.
\end{equation*}
From the proposition~\ref{prop:Z}, the partition function $Z_n$ can be written explicitly as 
\[
Z_{\beta,n} = \frac{1}{2} \sum_{m=1}^n  e^{(\beta \alpha - h)m}  +  e^{-hn}.
\]
Then we have that
\begin{eqnarray*}
\mu_{\beta,n}(C_k) &=& \frac{  \frac{1}{2} \sum_{m=k+1}^n  e^{(\beta \alpha - h)m}  + e^{-hn} }{  \frac{1}{2} \sum_{m=1}^n  e^{(\beta \alpha - h)m}  +  e^{-hn} }
\\
&=& 1- \frac{  \frac{1}{2} \sum_{m=1}^k  e^{(\beta \alpha - h)m}  }{  \frac{1}{2} \sum_{m=1}^n  e^{(\beta \alpha - h)m}  +  e^{-hn} }. 
\end{eqnarray*}

Next, for the moment, let us focus in obtaining a bound for the quotient of geometric series in the previous equation. First of all, for the sake of clarity, let us denote by $b$ the difference $\beta \alpha -h$ and observe that 
\[
 \sum_{m=1}^k  e^{(\beta \alpha - h)m}   =  \frac{e^{b (k+1)} - e^b }{e^b - 1} = \frac{e^{b k} - 1}{ 1- e^{-b}} < \frac{1}{ 1- e^b} e^{b k},
\]
where the inequality is valid for all $k\geq 1$ and all $b>0$, which is actually the case, since $\beta>\beta_{c}$. Analogously it is also clear that,
\[
 \sum_{m=1}^n  e^{(\beta \alpha - h)m}   = \frac{e^{b n} - 1}{ 1- e^{-b}}.
\]
Letting $K_1 := \frac{1}{1 - e^{-(\beta \alpha -h)}}$, we have that 
\begin{eqnarray*}
\frac{  \frac{1}{2} \sum_{m=1}^k  e^{(\beta \alpha - h)m}  }{  \frac{1}{2} \sum_{m=1}^n  e^{(\beta \alpha - h)m}  +  e^{-hn} } 
&< & \frac{K_{1}   e^{b k} }{ K_{1} (e^{bn} - 1) + 2e^{-hn}}
\\
&=&  \frac{K_{1}   e^{- b (n-k)} }{ K_{1} ( 1 -e^{-bn}) + 2 e^{-(h+b)n}}.
\end{eqnarray*}
Let us consider, the function
\[
g(n) := \frac{K_{1}  }{ K_{1} ( 1 -e^{-bn}) + 2 e^{-(h+b)n}},
\]
it is not difficult to see that $g(n)$ is non negative and bounded above for all $n\in \mathbb{N}$ and $b>0$. The value of $n$ minimizing $g$ is non-trivial (it is not necessarily $n=1$ or $n = \infty$) because $g$ varies non-monotonically with $n$. Thus, let us denote by $J_\beta$ the optimal upper bound for $g$, 
\[
J_\beta := \sup_{n\in \mathbb{N}} \bigg\{\frac{K_{1}  }{ K_{1} ( 1 -e^{-bn}) + 2 e^{-(h+b)n}} \bigg\}.
\]
Thus we have that
\[
\frac{  \frac{1}{2} \sum_{m=1}^k  e^{(\beta \alpha - h)m}  }{  \frac{1}{2} \sum_{m=1}^n  e^{(\beta \alpha - h)m}  +  e^{-hn} } 
<  J_\beta e^{-b (n-k)} = J_\beta e^{- (\beta \alpha -h)(n-k)},
\]
for every $n>k$ and all $b>0$ (that is, for $\beta > \beta_c := h/\alpha$). This result allows us to see that
\begin{eqnarray*}
\hspace{1cm}
\fl
\mu_{\beta,n}(C_k) 
&=& 1- \frac{  \frac{1}{2} \sum_{m=1}^k  e^{(\beta \alpha - h)m}  }{  \frac{1}{2} \sum_{m=1}^n  e^{(\beta \alpha - h)m}  +  e^{-hn} }
> 1- J_\beta e^{-(\beta \alpha -h) (n-k)}.
\end{eqnarray*}
The above result implies that, if we require that $\mu_{\beta,n}(C_k) > 1-\epsilon$ for a given $\epsilon$ it is necessary that 
\[
n > N := k+ \bigg| \frac{\log (\epsilon/J_\beta) }{\beta \alpha -h } \bigg|,
\] 
which proves the lemma. 
\qed

\subsection{Proof of the corollary~\ref{corollary:mass-concentration}}
Let us consider the set $ C_{L}  :=  \left(\bigcup_{k=1}^{L} B_k \right)^\mathrm{c}$, for $L\in\mathbb{N}$.
Now take the measure of $C_L^{\mathrm{c}}$ with respect to $\mu_{\beta,m+L}$ (we remind that $\beta>\beta_{c}$),
\[
\hspace{-2cm}
\mu_{\beta,m+L}\left( C_{L}^\mathrm{c} \right) = \mu_{\beta,m+L}\left( \bigcup_{k=1}^{L} B_k \right) = \sum_{k=1}^L \mu_{\beta,m+L} (B_{k}) =  1- \mu_{\beta,m+L}\left( C_{L}   \right).
\]
Since all the sets $B_{k}$ are pairwise disjoint. By Lemma~\ref{lemma:mass-concentration} we have that for every $\epsilon>0$ there is a $N^\prime$ such that
\begin{equation}
\label{ineq:mu-C-L}
\epsilon > 1- \mu_{\beta,m+L} (C_{L}) =  \sum_{k=1}^L \mu_{\beta,m+L} (B_{k}),
\end{equation}
for all $m+L >N^\prime$ and every $L\in \mathbb{N}$.  The proof of Lemma~\ref{lemma:mass-concentration} states that the value of $N^\prime$ is given by
\[
N^\prime := L + \bigg| \frac{\log (\epsilon/J_\beta) }{\beta \alpha -h } \bigg|.
\]
Therefore we have that the inequality~(\ref{ineq:mu-C-L}) is valid for all $m$ such that
\[
m+L > L +  \bigg| \frac{\log (\epsilon/J_\beta) }{\beta \alpha -h } \bigg|.
\]
This means that choosing $N$ as
\[
N:=\bigg| \frac{\log (\epsilon/J_\beta) }{\beta \alpha -h } \bigg|,
\]
then the corollary is true, concluding the proof.
\qed

\subsection{Proof of lemma~\ref{lemma:averages-f}}

First let us observe that for every $n>m$, the set $B_{n}\subset C_{m}$, and clearly $C_{m}\subset C_{m-1}$. By definition the sets $C_{m-1}$ and also $B_{m}$ both have $2^{m}$ subintervals. We will enumerate them as follows, $C_{m-1,i}$ and $B_{m,i}$ where $i = 1, \ldots, 2^{m-1}$. Notice that the analogous decomposition into $2^{n-1}$ subintervals can be achieved for the set $B_{n}$ (instead of $B_m$).

Next, let us for the moment consider the integral of $f(x)$ over the set $B_{m}$, which because of the disjointness of the subintervals $B_{m,i}$'s, one has that
\[
\frac{1}{\ell(B_{m})}\int_{B_{m}}f(x)dx = \frac{1}{\ell(B_{m})}\sum_{i=1}^{2^{m}}\int_{B_{m,i}}f(x)dx.
\]
Now, given that $f$ is continuous and the integral is taken with respect to Lebesgue, there is no measure contribution at the end points of every subinterval, so we can take the integral over $\overline{B_{m}}$, the closure of $B_{m}$, and thus we can make use of the mean value theorem for integrals, by which for every $j$, there exists a $y_{m,j}^{*}\in \overline{B_{m}}$ such that
\[
\hspace{-1.5cm}
\frac{1}{\ell(B_{m})}\int_{B_{m}}f(x)dx= \frac{1}{\ell(B_{m})}\sum_{i=1}^{2^{m}}\int_{\overline{B_{m,i}}}f(x)dx = \frac{1}{\ell(B_{m})}\sum_{i=1}^{2^{m-1}}f(y^{*}_{m,j})\ell(\overline{B_{m,j}}).
\]
One can write an analogous expression for the integral over $B_{n}$. Next, for every $i=1, \ldots, 2^{m-1}$, let us choose arbitrarily but then fixed, a $x_{i}^{*}\in C_{m-1,i}$, so we can write,

\begin{eqnarray*}
\hspace{-1.5cm}
\Delta :=  \Big| \frac{1}{\ell(B_{n})}\int_{B_{n}}f(x)dx - \frac{1}{\ell(B_{m})}\int_{B_{m}}f(x)dx \Big|  \\
\hspace{-1cm}
=  \Big| \frac{1}{\ell(B_{n})}\sum_{k=1}^{2^{n-1}}f(z^{*}_{n,k})\ell(\overline{B_{n,k}}) - \frac{1}{2^{m-1}}\sum_{i=1}^{2^{m-1}}f(x_{i}^{*}) + \frac{1}{2^{m-1}}\sum_{i=1}^{2^{m-1}}f(x_{i}^{*}) 
\\ \hspace{5.5cm} - \frac{1}{\ell(B_{m})} \sum_{i=1}^{2^{m-1}}f(y^{*}_{m,i})\ell(\overline{B_{m,i}})\Big|.
\end{eqnarray*}
Which, by explicitly  write the length of the sets, is equals to:
\[
\hspace{-2.5cm}
\Delta = \Big| \frac{1}{2^{n-1}}\sum_{k=1}^{2^{n-1}}f(z_{n,k}^{*}) - \frac{1}{2^{m-1}} \sum_{i=1}^{2^{m-1}}f(x_{i}^{*}) +
\frac{1}{2^{m-1}} \sum_{i=1}^{2^{m-1}}f(x_{i}^{*}) - \frac{1}{2^{m-1}}\sum_{i=1}^{2^{m-1}}f(y_{m,i}^{*})\Big|.
\]
By the triangle inequality,
\begin{equation}\label{Delta}
\hspace{-2.5cm}
\Delta \leq \Big| \frac{1}{2^{n-1}}\sum_{k=1}^{2^{n-1}}f(z_{n,k}^{*}) - \frac{1}{2^{m-1}} \sum_{i=1}^{2^{m-1}}f(x_{i}^{*})\Big| +
\Big| \frac{1}{2^{m-1}} \sum_{i=1}^{2^{m-1}}f(x_{i}^{*}) - \frac{1}{2^{m-1}}\sum_{i=1}^{2^{m-1}}f(y_{m,i}^{*})\Big|.
\end{equation}
We continue by estimating an upper bound for each part. Take first the second one, and let $K$ be the Lipschitz constant of $f$ so one has that,
\begin{eqnarray*}
\hspace{-2cm}
\Big| \frac{1}{2^{m-1}} \sum_{i=1}^{2^{m-1}}f(x_{i}^{*}) - \frac{1}{2^{m-1}}\sum_{i=1}^{2^{m-1}}f(y_{m,i}^{*})\Big| &\leq &
\frac{1}{2^{m-1}}\sum_{i}\big| f(x_{i}^{*} - f(y_{m,i}^{*})\big| \\
&\leq &\frac{1}{2^{m-1}} \sum_{i} K |x_{i}^{*} - y_{m,i}^{*}| \leq \frac{K}{3^{m-1}}.
\end{eqnarray*}
It remains to find an upper bound for the first part of equation (\ref{Delta}). Observe that we may write it as follows,
\begin{eqnarray*}
\hspace{-2cm}
\Big| \frac{1}{2^{n-1}}\sum_{k=1}^{2^{n-1}}f(z_{n,k}^{*}) - \frac{1}{2^{m-1}} \sum_{i=1}^{2^{m-1}}f(x_{i}^{*})\Big| = 
\\
\Big| \frac{1}{2^{m-1}}\sum_{i=1}^{2^{m-1}}\frac{1}{2^{n-m}}\sum_{j=(i-1)(2^{n-m})+1}^{i(2^{n-m})}f(z_{n,j}^{*}) - \frac{1}{2^{m-1}}\sum_{i=1}^{2^{m-1}}f(x_{i}^{*})\Big|.
\end{eqnarray*}
Notice that the summation over $j$ involves those values of $z_{n,j}^{*}$  belonging to the set $C_{m,i}$. This means that all the $z_{n,j}^{*}$ in the summation are as near to $x^*_i$ as $3^{-{m+1}}$ because $x_i^*\in C_{m-1,i}$. Now let us consider the part inside the sum over the $j$ index in the expression above. For each $i$ fixed, we define $\tilde{f}_{i}$ as the average function 
\[
\tilde{f}_{i} := \frac{1}{2^{n-m}}\sum_{j}f(z_{n,j}^{*}).
\]
Since $f$ is continuous in general, but specifically on every $C_{m-1,i}$, as a consequence of the Intermediate value theorem, there must exist a $\tilde{z}_{i}\in C_{m-1,i}$ such that $f(\tilde{z}_{i}) = \tilde{f}_{i}$. And thus, one may write, 
\begin{eqnarray*}
\hspace{-2.5cm}
\Big| \frac{1}{2^{m-1}}\sum_{i}\frac{1}{2^{n-m}}\sum_{j} f(z_{n,j}^{*}) - \frac{1}{2^{m-1}}\sum_{i} f(x_{i}^{*})\Big|  \leq
\frac{1}{2^{m-1}}\sum_{i}\Big| \sum_{j} f(z_{n,j}^{*}) - f(x_{i}^{*})\Big|  \\
\hspace{1cm}
=\frac{1}{2^{m-1}}\sum_{i}\Big| f(\tilde{z}_{i}) - f(x_{i}^{*})\Big| \leq \frac{1}{2^{m-1}}\sum_{i} K |\tilde{z}_{i} - x_{i}^{*}|\\ 
\hspace{1cm}
\leq \frac{1}{2^{m-1}} K \sum_{i}\max\{ |z_{n,j}^{*} - x_{i}^{*}|\}\leq \frac{K}{3^{m-1}}.
\end{eqnarray*}
Therefore, putting both parts together, one has that  $\Delta \leq \frac{2K}{3^{m-1}}$, so by choosing a $m$ sufficiently large one proves the lemma.
\qed

\subsection{Proof of lemma~\ref{lemma:mass-concentration-bar}}

 Given $n>m$ and $l = 1,\ldots, m$. Consider the measures $\mu_{\beta,n}\left( B_{n-l} \right)$ and $\mu_{\beta,m}\left( B_{m-l} \right)$ associated to the finite-step potential $n$ and $m$ of the sets $B_{n-l}$ and $B_{m-l}$, respectively. By the definition, one has that,
\[
\big| \mu_{\beta,n}\left( B_{n-l} \right) -  \mu_{\beta,m}\left( B_{m-l} \right)   \big| = 
\frac{1}{2}\bigg| \frac{e^{(\beta  \alpha -h)(n-l)}}{Z_{\beta,n}} - \frac{e^{(\beta  \alpha -h)(m-l)}}{Z_{\beta,m}} \bigg|.
\]
For the sake of brevity, let us define $b:=\beta  \alpha -h>0$, by hypothesis; since $\beta > \beta_c = h/\alpha$. So we may write that 
\begin{equation}
\big| \mu_{\beta,n}\left( B_{n-l} \right) -  \mu_{\beta,m}\left( B_{m-l} \right)   \big| = 
\frac{1}{2}\bigg| \frac{e^{b n}}{Z_{\beta,n}} - \frac{e^{bm}}{Z_{\beta,m}} \bigg| e^{-bl}.
\label{eq:diff-mus-ebl}
\end{equation}
Now, let  $K_1 := \frac{1}{1-e^{-b}}$, so we can also write the partition function $Z_{\beta,n}$ as follows:
\[
Z_{\beta,n} = \frac{1}{2} \sum_{k=1}^n e^{bk} + e^{-hn} = K_1 e^{bn} + e^{-hn}.
\]
Using the expression above we have that
\begin{eqnarray*}
\fl
\frac{1}{2}\bigg| \frac{e^{b n}}{Z_{\beta,n}} - \frac{e^{bm}}{Z_{\beta,m}} \bigg|  &=& 
\bigg|  \frac{e^{bn}}{ K_1 (e^{bn}-1) + 2 e^{-hn}} - \frac{e^{bm}}{ K_1 (e^{bm}-1) + 2 e^{-hm}} \bigg|
\nonumber
\\
\fl
&=&
\bigg|  \frac{1}{ K_1- K_1 e^{-bn} + 2 e^{-(h+b)n}} -  \frac{1}{ K_1- K_1 e^{-bm} + 2 e^{-(h+b)m}} \bigg|
\nonumber
\\
\fl
&=&
\bigg| \frac{K_1 (e^{-bn} -e^{-bm}) - 2(e^{- \beta \alpha n}- e^{- \beta \alpha m })}{\left( K_1- K_1 e^{-bn} + 2 e^{- \beta \alpha n}\right)\left( K_1- K_1 e^{-bm} + 2 e^{- \beta \alpha m} \right)} \bigg|.
\end{eqnarray*}

Next we define the function $g_1 : \mathbb{N} \to \mathbb{R}$ as follows
\[
g_1(k) :=  K_1- K_1 e^{-b k } + 2 e^{- \beta \alpha k},
\]
which is a  bounded above non-negative function. The latter means that $g_1$ attains an infimum, which we denote by $ G_1$, such that $0<  G_1 < \infty$. So we are able to give the following estimate
\begin{eqnarray*}
\frac{1}{2}\bigg| \frac{e^{b n}}{Z_{\beta,n}} - \frac{e^{bm}}{Z_{\beta,m}} \bigg|  
&\leq&
\frac{1}{G_{1}^{2}}
\big| G_1 (e^{-bn} -e^{-bm}) - 2(e^{-\beta \alpha n}- e^{- \beta \alpha m }) \big|
\\
&\leq& 
\frac{1}{G_{1}^{2}}
K_1 \big| e^{-bn} -e^{-bm}\big| + 2\big|e^{- \beta \alpha n}- e^{- \beta \alpha m } \big|
\\
&\leq& 
\frac{1}{G_{1}^{2}}
(K_1  + 2)\big| e^{-bn} -e^{-bm}\big|,
\end{eqnarray*}
since $b = \beta \alpha - h$. Therefore, from~(\ref{eq:diff-mus-ebl}), the last inequality and making the sum over the index $l$, one has
\begin{eqnarray*}
\fl
\sum_{l=0}^{L-1} \big| \mu_{\beta,n}\left( \overline{C}_{n-l} \right) - \big| \mu_{\beta,m}\left( \overline{C}_{m-l} \right)   \big| e^{-bl} 
&\leq& \sum_{l=0}^{L-1}
\frac{1}{G_{1}^2} (K_1  + 2)\big| e^{-bn} -e^{-bm}\big|  e^{-bl}
\\
&\leq&
\frac{1}{G_{1}^2} (K_1  + 2)\left(\big| e^{-bn}\big| + \big| e^{-bm}\big|\right) \left( \frac{1-e^{-bL}}{1-e^{-b}}\right)
\\
&\leq&
\frac{ 2(K_1  + 2)}{G_{1}^2}K_1 e^{-bm}.
\end{eqnarray*}
Observe from the previous expression that taking $m$ larger than the following quantity, which we define by $N$, 
\[
N :=  \left \lceil{\bigg|\log \left( \frac{\epsilon G_{1}^{2}K_1 }{2 (K_1  + 2)}\right)\bigg|}\right \rceil ,
\]
then one assures that 
\[
\sum_{l=0}^{L-1} \big| \mu_{\beta,n} \left( B_{n-l} \right)  -\mu_{\beta,m} \left(B_{m-l} \right)  \big| < \epsilon,
\]
which proves the lemma. 
\qed

\noindent
\subsection{Proof of lemma~\ref{lemma:residual}}
By definition of the measure $\mu_{\beta,n}(C_n) $ of the set $C_{n}$, we have that, 
\[
\mu_{\beta,n}(C_n)   = \frac{ e^{-h n}}{Z_{\beta,n}}.
\]
Since $Z_{\beta,n}$ can be written explicitly as,
\[
Z_{\beta,n} = K_1 e^{(\beta \alpha -h )n} + e^{-hn},
\]
(cf.~proof of Lemma~\ref{lemma:mass-concentration-bar}) then it is clear that
\[
\mu_{\beta,n}(C_n)   = \frac{e^{-hn}}{ K_1 e^{(\beta \alpha-h)n} + e^{-hn}} =  \frac{1}{ K_1 e^{\beta\alpha  n} +1} < \frac{e^{-\beta \alpha n}}{K_1}.
\]
Given $\epsilon$, one can take, 
\[
N:=-\frac{\log \left( \epsilon K_1\right)}{\beta \alpha}  
\]
and thus, for all $n> N $, we have that 
\[
\mu_{\beta,n}(C_n)<\epsilon
\]
which proves our claim.
\qed
\\


\nocite{*}

\bibliography{fractal_refs.bib} 


\end{document}